\documentclass[12pt]{article}%
\usepackage{amssymb}
\usepackage{hyperref}
\usepackage[round]{natbib}
\usepackage{amsfonts}
\usepackage{amsmath}
\usepackage[nohead]{geometry}
\usepackage[singlespacing]{setspace}
\usepackage[bottom]{footmisc}
\usepackage{indentfirst}
\usepackage{endnotes}
\usepackage{graphicx}%
\usepackage{enumerate}
\usepackage{mathrsfs}
\usepackage{rotating}
\newtheorem{theorem}{Theorem}

\newtheorem{lemma}[theorem]{Lemma}

\newtheorem{proposition}{Proposition}

\newenvironment{proof}[1][Proof]{\noindent\textbf{#1.} }{\ \rule{0.5em}{0.5em}}

\makeatletter
\def\@biblabel#1{\hspace*{-\labelsep}}
\makeatother
\DeclareMathOperator*{\argmax}{arg\,max}

\begin{document}
\epstopdfsetup{outdir=./}
\title{Efficient Network Structures with Separable Heterogeneous Connection Costs\footnote{This is the pre-print version of the following article: \textit{B. Heydari, M. Mosleh, and K. Dalili, “Efficient network structures with separable heterogeneous connection costs,” Economic
Letters Journal, vol. 134, no. September, pp. 82–85, 2015}., which has been published in final form at \href{http://dx.doi.org/10.1016/j.econlet.2015.06.014}{doi:10.1016/j.econlet.2015.06.014.} }}

\author{ Babak Heydari\textsuperscript{a} \thanks{Address: 1 Castle Point Terrace,
Hoboken, NJ 07030, USA, e-mail:
\textit{babak.heydari@stevens.edu}, web: http://web.stevens.edu/cens/}, Mohsen Mosleh\textsuperscript{a}, Kia Dalili\textsuperscript{b} \medskip\\{\normalsize \textsuperscript{a}School of Systems and Enterprises, Stevens Institute of Technology, Hoboken, NJ 07030}\medskip\\{\normalsize \textsuperscript{b}Facebook Inc, New York, NY 10003}}


\maketitle

\sloppy

\onehalfspacing

\begin{abstract}

We introduce a heterogeneous connection model for network formation to capture the effect of cost heterogeneity on the structure of efficient networks. In the proposed model, connection costs are assumed to be separable, which means the total connection cost for each agent is uniquely proportional to its degree. For these sets of networks, we provide the analytical solution for the efficient network as a function of connection costs and benefits. We show that the efficient network exhibits a core-periphery structure. Moreover, for a given link density, we find a lower bound for the clustering coefficient of the efficient network, and compare it to that of the Erd\H{o}s-R\'enyi random networks.

\end{abstract}

\strut


\strut


\pagebreak
\doublespacing 

\section{Introduction}
\label{intro} 

Network formation models are increasingly being used in a variety of economic contexts and other multi-agent systems. These models often study the structural conditions of efficiency, network social
welfare, and stability, which is a measure of individual incentives to form, keep or sever links (\cite{jackson2008social}).

We build our model based on the \textit{Connection model} proposed by \cite{jackson1996strategic}\footnote{\cite{jackson1996strategic} developed their model based on the notion of pairwise stability or \textit{two-sided link formation} where a link is formed upon the ``mutual consent" of two agents. There is also another line of literature from \cite{bala1997self,goyal1993sustainable} that studies one-sided and non-cooperative link formation, where agents unilaterally decide to form the links with another agent.} in which agents can benefit from both direct and indirect connections, but only pay for their direct connections. Benefits of indirect connections generally decrease with distance. \cite{jackson1996strategic} demonstrated that, for the homogeneous case,  the efficient network can only take one of three forms: a complete graph, a star or an empty graph depending on connection cost and benefits. Several models have been proposed to introduce heterogeneity into the connection model, (see for instance: \cite{galeotti2006network,jackson2005economics,persitzcore,vandenbossche2013network}); the focus has mainly been on conditions for stability, with few references to efficiency. Finding general analytical solutions for the efficient networks with heterogeneous costs can be intractable, 
see for example \cite{carayol2009knowledge}.



Here, we focus on finding efficient networks for a particular model of cost heterogeneity  that we refer to as the \textit{separable connection cost} model, in which shares of nodes' costs from each connection are heterogeneous, yet fixed and independent of to whom they connect. This is motivated by networks in which heterogeneous agents are each endowed with some resources (time, energy, bandwidth, etc) and the total resource needed to establish and maintain connections for each node can be approximated to be proportional to its degree. We further assume homogeneous benefits decaying with distance. We provide an exact analytical solution for efficient connectivity structures under these assumptions and show that such networks have at most one connected component, exhibit a core-periphery structure and have diameters no larger than two. We further provide a lower bound for the clustering coefficient and discuss the pairwise stability implications of efficient networks.

\section{Model}
\label{model}
For a finite set of agents $N=\{1, \dots, b\}$, let $b:\{1, ... ,n-1\} \rightarrow \mathbb{R}$ represent the benefit that an agent receives from (direct or indirect) connections to other agents as function of the distance between them in a graph. Following \cite{jackson1996strategic}, the (distance-based) utility function of each node, $u_i(g)$, in a graph $g$ and the total utility of the graph, $U(g)$, are as follows:
\begin{equation} 
\label{utility_connection_model}
\begin{split}
u_i(g)&=\sum_{j\neq i: j\in  N_i^{n-1}(g)} b(d_{ij}(g))-\sum_{j\neq i: j\in N_i(g) } c_{ij}\\
U(g)&=\sum_{i=1}^{n}u_i(g)
\end{split}
\end{equation} 

where $N_i(g)$ is the set of nodes to which $i$ is linked, and $N^k_i(g)$ is the set of nodes that are path connected to $i$ by a distance no larger than $k$. $d_{ij}(g)$ is the distance between $i$ and $j$, $c_{ij}$ is the cost that node $i$ pays for connecting to $j$, and $b$ is the benefit that node $i$ receives from a connection with another node in the network. We assume that $b(k)>b(k+1)>0$ for any integer $k \geq 1$.

Let complete graph $g^N$ denotes the set of all subsets of $N$ of size 2. The network $\tilde{g}$ is 
efficient, if $U(\tilde{g}) \geq U(g^\prime)$ for all $g^\prime \subset g^N$, which indicates that $\tilde{g}=\argmax_g\sum_{i=1}^{n}u_i(g)$. Assuming the \textit{separable} cost model as introduced earlier, connection costs in Equation~\ref{utility_connection_model}, for a link between $i$ and $j$ can be written as $c_{ij} = c_i$, $c_{ji} = c_j$.
We then introduce a connection cost vector, $C$, and  without loss of generality rename nodes such that $c_1<c_2<...<c_n$. 
%


\begin{figure*}[t]
\centering
\includegraphics[width=2in]{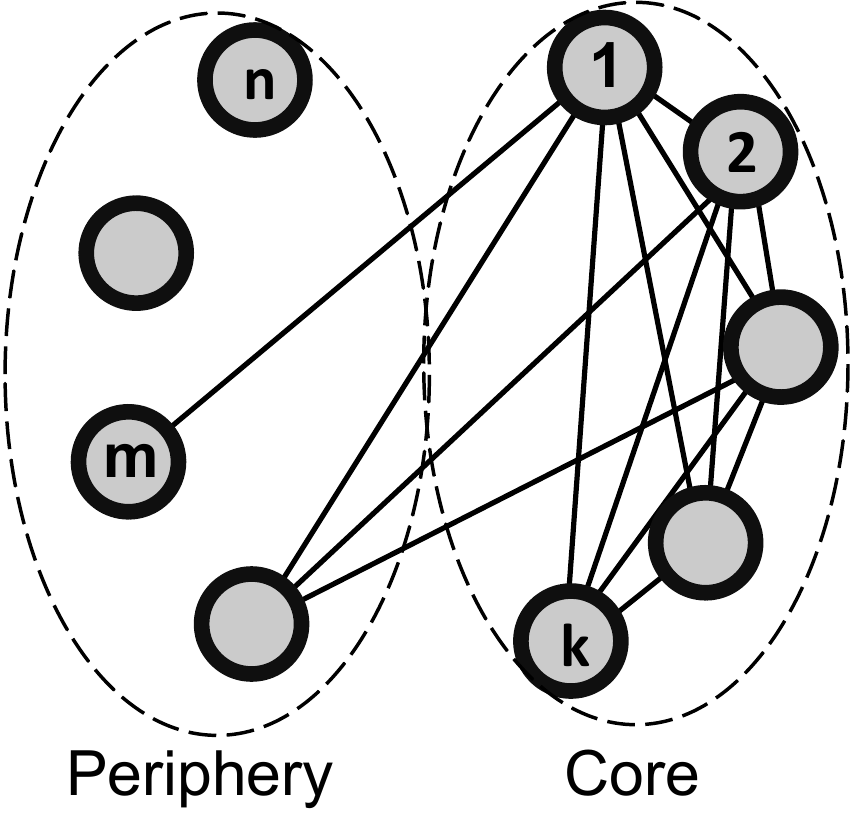}
\caption{Structure of the efficient network with separable heterogeneous connection model. The connected component has a \textit{generalized star} structure. $n$ is the number of nodes in the network, $m$ is the number of nodes in the connected component, and $k$ is the number of nodes in the core (complete subgraph).  Nodes are sorted according to their costs i.e. $c_1 \leq c_2 \leq \cdots \leq c_n$. }
\label{figure1}
\end{figure*}

\subsection{Efficient Structures under Separable Connection Costs}
In Lemma~\ref{lemma2}, we determine the efficient structure for a connected component. Then in Proposition~\ref{personal_cost_prop}, we determine the structure of the efficient network in general.



\begin{lemma}
\label{lemma2}
If the efficient network with separable cost model is connected then it has a ``generalized star" structure with the following characteristics: (a) All nodes are connected to node 1 ( the node with minimum connection cost). (b) Nodes $i$ and $j$ ($i,j\neq 1$) are connected iff $b(1)-b(2)>.5(c_i+c_j)$.

\end{lemma} 
\begin{proof}
Let $N$ represents nodes in the connected network. If there exists a subset of nodes $M=\{v_1, ... ,v_m\}$ $(M \subset N)$ that are not connected to node 1, we show that the network is not efficient. Since $N$ is connected, there exists a set  of links, $L=\{l_1, ..., l_m\}$  where $l_i$ is adjacent to $v_i$\footnote{For any subset $M=\{v_1,\cdots , v_m\}$ of a connected network $N$ ($M\not\equiv N$), we can show that there are links $L=\{l_1,\cdots,l_m\}$ in $N$ such that $v_i$ is adjacent to $l_i$.}. Suppose $l_i$ connects $v_i$ to $w_i$ ( $w_i\neq 1$ by definition). Now, if we remove all $l_i$s and connect all $v_i$s to node 1, we have reduced the total connection cost of the network by $ m c_1 - \sum\limits_{k=1}^{m} c_{v_k} < 0$.

Now, to address the benefits, note that we have not changed the number of links; therefore direct benefits remain the same. Furthermore, the diameter of the new network is 2. Therefore every distance that is not 1 is capped at 2, making the total benefit larger than that of the original network. This results in at improvement in the total utility, indicating that the original network was not efficient.

Furthermore, having established that the maximum distance in the efficient network is no larger than two, every node $i$ and $j$ ($i,j \neq 1$) are connected iif $b(1)-b(2)>.5(c_i+c_j)$.


\end{proof}

Proposition~\ref{personal_cost_prop} determines the structure of the efficient network and shows that the efficient network is a spectrum of solutions. 

\begin{proposition}
\label{personal_cost_prop}
In the connection model, for a finite set of agents, $N= \{1,..,n\}$,  if $c_{ij} = c_i$ for all $i,j \in N$, where $c_i \in C=\{c_1, c_2, ... , c_n\} $ and assuming, $c_1 < c_2 < \cdots < c_n$, the structure of the efficient network is as follows: Let $m$ be the largest integer between $1$ and $n$ such that $2b(1)+2(m-2)b(2)>(c_m+c_1)$. If $i > m$, then $i$ is isolated.  If $i  \leq m$, then there is exactly one link between $i$ and $1$; also there is one link between $i$ and $j$ ($1 < i,j \leq m$) iff $b(1)-b(2)>.5(c_i+c_j)$.

\end{proposition}
\begin{proof}

First, we show there is at most one connected component in the efficient network. Next, we find the condition for each node to be in the connected component, which has the generalized star structure according to Lemma~\ref{lemma2}.

Assume that the efficient network has more than one (e.g. two) connected components with $(m_i,\ell_i)$ being respectively the number of nodes and links in connected component $i$. According to Lemma 1, each connected component has a \textit{generalized star} structure.
The total benefit of each component is  $B_1=2\ell_1 b(1)+(m_1(m_1-1)-2\ell_1)b(2)$  and $B_2=2\ell_2 b(1)+(m_2(m_2-1)-2\ell_2)b(2)$ respectively. Suppose $h$ and $h^\prime$ are nodes with minimum costs in component 1 and 2 respectively and without loss of generality $c_h < c_{h^\prime}$. If we disconnect all links connected to $h^\prime$ and connect them directly to $h$, total cost decreases by $(c_{h^\prime}-c_h)$ per link. This also results in the total benefit $B=(\ell_1+\ell_2) b(1)+( (m_1+m_2)(m_1+m_2-1)-2(\ell_1+\ell_2))b(2)$, which is strictly greater than $(B_1+B_2)$. 

To determine which nodes belong to the connected component, $G_C$, in the efficient network, we define, for node $i$, $A_i \triangleq 2b(1)+2(k-2)b(2) - c_1- c_i$ where $k$ is the number of nodes in the connected component $G_C$. We show that $i$ is in $G_C$ \textit{iff} $A_i \geq 0$. First,  $A_i  > 0$ is the sufficient condition for $i$ to be in $G_C$. This is because connecting $i$ to node $1$ increases the total utility by exactly $A_i$, as the diameter of $G_C$ is at most 2, according to Lemma 1. Also if $A_i <0$, then $i$ will be isolated so $A_i \geq 0$ is also the necessary condition. This is because $i$ cannot be only connected to 1  since $A_i <0$, so the only way for $i$ to be connected is by having more than one link. From Lemma~\ref{lemma2}, for $i$ to have a link to  $j \neq 1$, we must have $c_i+c_j < 2b(1) - 2b(2)$. But: $c_i+c_j > c_1+c_i > 2b(1) + 2(k-2)b(2) > 2b(1)-2b(2)$, where we use the fact that $c_j>c_1$ and $A_i <0$, so $i$ cannot have more than one connection either, thus $i$ will be isolated. Note that $A_i  > 0$ also means $A_j > 0$ for all $j <i$ since $c_j < c_i$, thus all lower cost nodes will also be in $G_C$, so the smallest $i$ for which $A_i < 0$ provides the size of the connected component in the efficient network.

\end{proof}

A typical structure for the efficient network with heterogeneous separable cost model is illustrated in Figure~\ref{figure1}.

\subsection{Characteristics of Networks with Separable Cost Model}

\subsubsection{Core-Periphery structure}

We show that the efficient network has a Core-periphery structure, a widely observed structure in various social and economic networks (i.e. see for example \cite{zhang2014identification, rombach2014core}). We adopt the formal definition from \cite{bramoulle2007anti}, which states that a graph $g$ has a core-periphery structure when agents can be partitioned
into two sets, the \textit{core} $C$ and the \textit{periphery} $P$, such that all partnerships are formed within the core
and no partnership is formed within the periphery. For an efficient network, let $k$ be the largest integer between $2$ and $n$ such that $b(1)-b(2)>.5(c_{k-1}+c_{k})$. The efficient network can be partitioned into a set $C=\{1, \dots , k\}$, which forms a complete subgraph and a set $P=\{k+1, \dots , n\}$ which can only have connections to the complete subgraph. If $b(1)-b(2)>(c_{k-1}+c_{k})$ then $k$ and $k-1$ are connected and there is also a link between every node $i , j (i,j \leq k$ and $c_i,c_j \leq c_k)$, which forms a complete subgraph. Similarly, we can show that for every $i \in P$, which is connected to a node $j$, $c_j\leq c_k$ and $j\not \in P$.

\subsubsection{Clustering coefficient}

Unlike efficient networks with homogeneous cost model whose clustering coefficients are either one (complete graph) or zero (star structure or empty graph), the clustering coefficient of the efficient network with heterogeneous costs covers a wide range. 
To find a lower bound, we find the minimum global clustering coefficient of all the efficient networks with a given link density and various connection cost values. To this end, we construct a network that does not break the efficiency condition, provided in the previous section, while producing the minimum possible number of triangles for the given density as follows: Starting from an empty graph and node $k=1$, we only establish links from node $k$ to every node $i>k$ in ascending order. We repeat this process for $k=1, \dots, p+1$ until the total number of links reaches $\ell$. $p$ is the largest integer such that $\ell=\sum_{k=1}^{p}(n-k)+J$, where $J$ is the residual number of links in the $p+1$ round of incomplete iteration $(J<n-p+1)$. At the $k$-th iteration, the total number of connected triplets is increased by ${n-k \choose 2}+2(k-1)(n-k)+{J \choose 2}+2Jp$. The first term in this equation results from the fact that node $k$ will act as a local hub with $(n-k)$ new links who provide ${n-k \choose 2}$ new triplets. The second term is the additional triplets resulting from new triangles formed in the $k$-th round since there are $(k-1)(n-k)$ triangles formed in round $k$ which add 3 triplets, one of which have already been counted in previous iterations. The remaining two terms follow a similar logic for the residual links ($J$). Similarly, the total number of triangles formed at $k$-th iteration is increased by $(k-1)(n-k)+pJ$. Therefore, for a given number of links ($\ell$), we have:

\begin{equation} 
\label{clustering}
\begin{split}
C^{min}_{eff}(g,\ell)&=\frac{3\times\mbox{number of triangles}}{\mbox{number of connected triplets of nodes}}\\
&=\frac{3\times\{\sum_{k=1}^{p}(k-1)(n-k)+pJ\} }{\sum_{k=1}^{p}\{{n-k \choose 2} +2(k-1)(n-k)\}+{J \choose 2}+2Jp}\\
\end{split}
\end{equation} 
Compared to an Erd\H{o}s-R\'enyi (ER) random network with the same link density, some algebra reveals that the clustering clustering coefficient of the efficient heterogeneous network exceeds that of an ER network when then number of nodes is greater than 10 and the link density is larger than $\frac{4(2n-5)}{n(n-1)}$. For sufficiently large networks, the condition simplifies to having link density greater than $\frac{8}{n}$.

\section{Conclusion} 
Heterogeneous yet separable connection costs cover important classes of real networks.
We showed that efficient structures for such networks can be solved exactly and have diameter no larger than two; we also discussed the transitivity and core-periphery nature of such networks. 
Although benefits are still assumed to be homogeneous, one can easily take into account heterogeneity of direct benefits, as long as the separability assumption is maintained, i.e. cost and direct benefit terms appear together in all analysis and cost terms can capture heterogeneity of direct benefits by embedding them as an off-set in the fixed costs of nodes.

We mainly focused on the notion of efficiency for the networks with separable connection cost model. Further studies can investigate stability for the proposed model, and find conditions under which stable and efficient structures coincide. Moreover, there are cases where the separable cost assumption does not hold, for example when the link cost is the function of how \textit{similar} the two nodes are. These cases in general can be intractable and approximate methods such as \textit{island models} as discussed in \cite{jackson2005economics} can be used.



\section*{Acknowledgements} 
This work was supported by DARPA Contract NNA11AB35C. The authors are grateful to Peter Ludlow (Stevens) and Pedram Heydari (UCSD) for insightful comments. 
\singlespacing


\footnotesize

\normalsize

\bibliographystyle{apalike}

\bibliography{workingpaper}
 
%
%
%
%
%
%

\end{document}